\documentclass[11pt,a4paper]{article}
\usepackage{amssymb,amsmath,amsfonts,amsthm,amstext,amscd,array}
\usepackage{mathrsfs}
\usepackage{hyperref}
\usepackage{pdfsync}
\usepackage{bbm}
\usepackage{bm}
\usepackage[arrow,matrix,curve]{xy}
\usepackage{bbding}
\usepackage{wasysym}

\usepackage{epsf}
\usepackage{epsfig}
\usepackage{wrapfig}
\usepackage{ytableau}
\input{epsf}

\def\nn{\nonumber}

\newcommand{\mbf}[1]{{\boldsymbol {#1} }}
\def\ii{{\,{\rm i}\,}}
\def\dd{{\rm d}}

\newcommand{\unit}{\mathbbm{1}}   			

\newcommand{\eq}{\begin{equation}}
\newcommand{\eqend}{\end{equation}}
\newcommand{\eqa}{\begin{eqnarray}}
\newcommand{\nonueqa}{\begin{eqnarray*}}
\newcommand{\eqaend}{\end{eqnarray}}
\newcommand{\nonueqaend}{\end{eqnarray*}}

\newcommand{\bma}[1]{\begin{array}{#1}}
\newcommand{\ema}{\end{array}}
\newcommand{\bc}{\begin{center}}
\newcommand{\ec}{\end{center}}

\newcommand{\newsection}{\setcounter{equation}{0}\section}

\newcommand{\complex}{{\mathbb C}} 
\newcommand{\real}{{\mathbb R}} 

\newcommand{\quat}{{\mathbb H}} 

\newif\ifold             \oldtrue

\def\nn{\nonumber}

\def\be{\begin{equation}}
\def\ee{\end{equation}}
\def\bea{\begin{eqnarray}}
\def\eea{\end{eqnarray}}
\def\bd{\begin{displaymath}}
\def\ed{\end{displaymath}}

\newcommand{\beq}{\begin{eqnarray}}
\newcommand{\eeq}{\end{eqnarray}}

\makeatletter
\newdimen\normalarrayskip              
\newdimen\minarrayskip                 
\normalarrayskip\baselineskip
\minarrayskip\jot
\newif\ifold             \oldtrue            
\def\arraymode{\ifold\relax\else\displaystyle\fi} 
\def\@arrayskip{\ifold\baselineskip\z@\lineskip\z@
     \else
     \baselineskip\minarrayskip\lineskip2\minarrayskip\fi}
\def\@arrayclassz{\ifcase \@lastchclass \@acolampacol \or
\@ampacol \or \or \or \@addamp \or
   \@acolampacol \or \@firstampfalse \@acol \fi
\edef\@preamble{\@preamble
  \ifcase \@chnum
     \hfil$\relax\arraymode\@sharp$\hfil
     \or $\relax\arraymode\@sharp$\hfil
     \or \hfil$\relax\arraymode\@sharp$\fi}}
\def\@array[#1]#2{\setbox\@arstrutbox=\hbox{\vrule
     height\arraystretch \ht\strutbox
     depth\arraystretch \dp\strutbox
     width\z@}\@mkpream{#2}\edef\@preamble{\halign \noexpand\@halignto
\bgroup \tabskip\z@ \@arstrut \@preamble \tabskip\z@ \cr}%
\let\@startpbox\@@startpbox \let\@endpbox\@@endpbox
  \if #1t\vtop \else \if#1b\vbox \else \vcenter \fi\fi
  \bgroup \let\par\relax
  \let\@sharp##\let\protect\relax
  \@arrayskip\@preamble}
\makeatother

\def\be{\beta}

\newtheorem{lemma}[equation]{Lemma}

\theoremstyle{definition}

\usepackage{hyperref}
\hypersetup{colorlinks,%
citecolor=black,%
filecolor=black,%
linkcolor=black,%
urlcolor=black}

\begin{document}
\begin{titlepage}

\begin{flushright}
\small
\baselineskip=12pt
MPP-2018-111
\end{flushright}
\normalsize

\begin{center}

\vspace{1cm}

\baselineskip=24pt

{\Large\bf Recurrence relations for symplectic realization of\\ (quasi)-Poisson structures}

\baselineskip=14pt

\vspace{1cm}

{\bf Vladislav G. Kupriyanov}
\\[5mm]
\noindent  {\it Max-Planck-Institut f\"ur Physik,
  Werner-Heisenberg-Institut\\ F\"ohringer Ring 6, 80805 M\"unchen, Germany
}
\\ and {\it CMCC-Universidade Federal do ABC, Santo Andr\'e, SP, 
Brazil}\\ and {\it 
Tomsk State University, Tomsk, Russia}
\\
Email: \ {\tt
    vladislav.kupriyanov@gmail.com}
\\[30mm]

\end{center}

\begin{abstract}
\baselineskip=12pt
\noindent
It is known that any Poisson manifold can be embedded into a bigger space which admites a description in terms of the canonical Poisson structure, i.e., Darboux coordinates. Such a procedure is known as a symplectic realization and has a number of important applications like quantization of the original Poisson manifold. In the present paper we extend the above idea to the case of quasi-Poisson structures which should not necessarily satisfy the Jacobi identity. For any given quasi-Poisson structure $\Theta$ we provide a recursive procedure of the construction of a symplectic manifold, as well as the corresponding expression in the Darboux coordinates, which we look in form of the generalized Bopp shift. Our construction is illustrated on the exemples of the constant $R$-flux algebra, quasi-Poisson structure isomorphic to the commutator algebra of imaginary octonions and the non-geometric M-theory $R$-flux backgrounds. In all cases we derive explicit formulae for the symplectic realization and the generalized Bopp shift. We also discuss possible applications of the obtained mathematical structures.

\end{abstract}

\end{titlepage}

\section{Introduction}

\emph{Symplectic realizations} are an important mathematical tool for investigation of Poisson and \emph{quasi}-Poisson manifolds with a number of important applications, from the study of classical dynamics \cite{KS18} and the algebra of symmetries \cite{Gracia-Bondia:2017fai}, to a quantization, see, e.g., \cite{KS18,Hammou:2001cc,GraciaBondia:2001ct,Pachol:2015qia}. To be more precise let us recall that a \emph{symplectic realization} of a Poisson structure $\omega$ on a
manifold $M$ is a symplectic manifold $(S,\Omega)$ together with a
surjective submersion ${\sf p}:S\to M$ which preserves the Poisson
structures: ${\sf p}_*\,\Omega^{-1} =\omega$. It is a fundamental
result in Poisson geometry that any Poisson manifold admits a
symplectic realization. The original local construction for
$M=\real^d$ is due to~\cite{Weinstein-local}; it proceeds by taking
$S=T^*M$ to be the phase space of $M$, with the canonical projection
${\sf p}:T^*M\to M$, and $\Omega$ to be the integrated pullback of the
canonical symplectic structure $\dd p_i\wedge\dd x^i$ on $T^*M$ by the flow of the vector field $\omega^{ij}(x)\, p_i\, \partial_j$, where $(x,p)\in T^*M=\real^d\times(\real^d)^*$. The early global constructions based on integrating symplectic groupoids are due to~\cite{Karasev,Weinstein-groupoid,Weinstein}. The extension to almost symplectic realizations of twisted Poisson structures is established globally by~\cite{CattaneoXu}.

In the previous work \cite{KS18} we have proposed the symplectic realization for the monopole algebra, in which the bracket of the covariant momenta is proportional to the magnetic field, $\{\pi_i,\pi_j\}=e\, \varepsilon_{ijk} \, B^k(\vec x\,)$, without requiring, $\vec\nabla\cdot\vec B=0$. The obtained formulation was used for studying the classical dynamics and quantization of the electric charge in a field of monopole distributions. The aim of the present paper is to extend the construction of \cite{KS18} to the case of an arbitrary quasi-Poisson structure.

Given an arbitrary bi-vector $\Theta=\frac12\, \Theta^{ij}(x) \, \partial_i\wedge\partial_j$ on a manifold $M$ of dimension $d$, the algebra of \emph{quasi-Poisson brackets}
\begin{equation}
\{x^i,x^j\}_Q=\alpha\, \Theta^{ij}(x) \ ,
\end{equation}
for local coordinates $x\in \real^d$ and a deformation parameter $\alpha\in\real$, is bilinear and antisymmetric but does not necessarily satisfy the Jacobi identity, i.e., in general 
\bea
\Pi^{ijk} = \mbox{$\frac{1}3$} \, \big(
\Theta^{il}\, \partial_l\Theta^{jk}+\Theta^{kl}\, \partial_l\Theta^{ij}+\Theta^{jl}\, \partial_l\Theta^{ki}
\big)\neq0 \ .\label{Pi}
\eea
To construct the symplectic realization one can ``double'' the local space to $\real^{2d}$ with coordinates $\zeta^\mu=(x^i,{\tilde x}_i)$ for $\mu=1,\dots,2d$ and construct a Poisson bracket
\begin{equation}
\{\zeta^\mu,\zeta^\nu\}_p= \Omega^{\mu\nu}(\zeta)=\Omega_0^{\mu\nu}+\alpha\, \Omega_1^{\mu\nu}(\zeta)+{\cal O}(\alpha^2)
\end{equation}
as a formal power series in the parameter $\alpha$, where $\Omega_0^{\mu\nu}$ is the canonical symplectic matrix. The Poisson brackets of the original coordinate functions are then
\begin{equation}
\{x^i,x^j\}_p = \alpha\, \omega^{ij}(x,\tilde x) \qquad \mbox{with} \quad \omega^{ij}(x,0)=\Theta^{ij}(x) \ .
\end{equation}
An important requirement is that if $\Theta$ is a Poisson bi-vector, i.e., (\ref{Pi}) vanishes, then $ \omega^{ij}(x,\tilde x)=\Theta^{ij}(x)$. In particular, it implies that,
$
 \omega^{ij}(x,\tilde x)=\Theta^{ij}(x)-\alpha\, \Pi^{ijk}(x)\, \tilde x_k +{\cal O}(\alpha^2) \ .
$
The expansion may be explicitly constructed by introducing local Darboux coordinates $\eta^\mu=(y^i,\pi_i)$ and writing the generalised Bopp shift
$
x^i=y^i-\mbox{$\frac\alpha2$}\, \Theta^{ij}(y)\, \pi_j+{\cal O}(\alpha^2)$, and $\tilde x_i=\pi_i \ .$

The paper is organized as follows. In the Section 2 we construct a symplectic realization of a Poisson structure $\Theta(x)=\omega(x)$. We derive a recurrence relation for the generalized Bopp shift and the corresponding symplectic structure $\Omega(\zeta)$. Then, in the Section 3 we generalize the construction of the Section to the case of quasi-Poisson structure $\Theta$. It involves the non-trivial equations for the construction of the tensor $ \omega^{ij}(x,\tilde x)=\Theta^{ij}(x)+{\cal O}(\alpha) $. We prove the existence of the solution and give the recurrence relation for its construction. In the Section 5 we discuss the exemples of a symplectic realization of the quasi-Poisson structures mainly related to the non-geometric backgrounds in string and M-theory. For all exemples we derive the explicit form of the symplectic sructure in the ``double'' space and the generalized Bopp shift. In appendix we briefly review the algebra of octonions.

Through the text we will use different notations for the brackets: $\{\cdot,\cdot\}_p$ denotes the arbitrary Poisson (satisfying Jacobi identity) bracket, $\{\cdot,\cdot\}$ stands for the canonical Poisson bracket, and $\{\cdot,\cdot\}_Q$ indicates the quasi-Poisson bracket, when the Jacoby identity can be violated. 

\section{Poisson structure}

As a warm-up we start with more familiar case of symplectic realizations of Poisson
manifolds. Some blocks of the construction of this Section will be used for generic quasi-Poisson
structures. Note that for the case of the two-dimensional Poisson manifold this problem was solved first in \cite{Gomes:2009tk}.

Suppose that coordinates $x^i,$ $i=1,..,N$ satisfy the algebra of a given Poisson brackets
\begin{equation}\label{1}
   \{x^i, x^j\}_p=\,\alpha \,\omega^{ij}(x),
\end{equation}
where $\alpha$ is a deformation parameter and $\omega^{ij}(x)$ is a Poisson bi-vector. The Jacobi identity for the algebra (\ref{1}) reads:
\begin{equation}\label{JI}
\{x^i,\{x^j,x^k\}_p\}_p+\{x^i,\{x^j,x^k\}_p\}_p+\{x^i,\{x^j,x^k\}_p\}_p=0.
\end{equation}
In this section we will describe a recursive procedure of the construction of a $2N$-dimensional symplectic manifold with coordinates
 $\left( x^{i},\tilde x_{i}\right) ,$ satisfying the algebra
\begin{eqnarray}
\{x^i, x^j\}_p&=&\alpha\, \omega^{ij}(x),  \label{2} \\
\{x^i, \tilde x_j\}_p&=&\delta^i_j(x,\tilde x)=\delta^i_j+\alpha\,{\delta^{(1)}}^i_j(x,\tilde x)+{\cal O}(\alpha^2),\nonumber\\
\{\tilde x_i, \tilde x_j\}_p&=&\alpha\,\varpi_{ij}(x,\tilde x),\nonumber
\end{eqnarray}
here $\delta^{i}_j$ is a Kronecker delta and functions $\delta^i_j(x,\tilde x)$ and $\varpi_{ij}(x,\tilde x)$ should be fixed from the condition that the complete algebra of Poisson brackets satisfy the Jacobi identity.

To solve the above problem we will use the generalized Bopp shift \cite{KV,Kup14}, i.e., express the original coordinates $x^i$, and the double ones $\tilde x_i$, in terms of the Darboux coordinates $\left( y^{i},\pi _{i}\right) $, satisfying the canonical Poisson brackets,
\begin{equation}\label{3}
    \{y^i, \pi_j\}=\delta^i_j,\,\,\,\{y^i, y^j\}=\{\pi_i, \pi_j\}=0.
\end{equation}
The expression for the coordinates $x^{i}$ we are looking in a form%
\begin{equation}
x^{i}=y^{i}+\sum_{n=1}^{\infty}\Gamma^{i\left( n\right) }\left( y\right)
\left( \alpha\pi\right) ^{n},  \label{4}
\end{equation}
where $\Gamma^{i\left( n\right) }\left( y\right)
=\Gamma^{ij_{1}...j_{n}}\left( y\right)$. Note that by the construction the coefficient functions, $\Gamma^{ij_{1}...j_{n}}\left( y\right)$ should be symmetric in the last $n$ indices. 

It is better to keep the totally symmetric
part of $\Gamma^{ij_{1}...j_{n}}\left( y\right)$ vanishing, since it implies the stability of unity see for more details \cite{KV}. This implies that under the permutations
of indices the tensor $\Gamma^{ij_{1}...j_{n}}\left( y\right)$ transforms according to the Young
tableau:
\begin{equation}
\ytableausetup{mathmode}
\begin{ytableau}
j_1 & j_2 & \none[\dots] & j_n \\
i \label{YGamma}
\end{ytableau}
\end{equation}

Up to the second order one writes:%
\begin{equation}
x^{i}=y^{i}+\alpha\,\Gamma^{ij}\left( y\right) \pi_{j}+\alpha^{2}\,\Gamma
^{ijk}\left( y\right) \pi_{j}\pi_{k}+{\cal O}\left( \alpha^{3}\right) .  \label{5}
\end{equation}
To find the tensors $\Gamma^{ij_{1}...j_{n}}\left( y\right) $ we substitute the expression (\ref{4}) in (\ref{1}) and obtain the equation
\begin{align}
& \left\{ y^{i}+\sum_{n=1}^{\infty }\Gamma ^{i\left( n\right) }\left(
y\right) \left( \alpha \pi \right) ^{n},y^{j}+\sum_{n=1}^{\infty }\Gamma
^{j\left( n\right) }\left( y\right) \left( \alpha \pi \right) ^{n}\right\} =
\label{6} \\
&\alpha\, \omega ^{ij}\left( y^{i}+\sum_{n=1}^{\infty }\Gamma ^{i\left(
n\right) }\left( y\right) \left( \alpha \pi \right) ^{n}\right) .  \notag
\end{align}%
Comparing the coefficients in the left and in the right-hand sides of (\ref%
{6}) in each order in $\alpha ,$ one obtains algebraic equations on the
coefficients $\Gamma ^{i\left( n\right) }\left( y\right) $ in terms of $%
\omega ^{ij}$ and lower order coefficients $\Gamma ^{i\left( m\right)
}\left( y\right) ,\ m<n$. In what follows we will prove by the induction the existence of the solution of these equations in all orders and also will provide the explicit recursive formulae for this solution.

In the first order in $\alpha $ we have from (\ref{6}):%
\begin{equation*}
\Gamma ^{[ji]}:=\Gamma ^{ji}-\Gamma ^{ji}=\omega ^{ij},
\end{equation*}%
with a solution $\Gamma ^{ij}=-\omega ^{ij}/2+s^{ij}$, where $s^{ij}$ is an
arbitrary symmetric matrix. Choosing $s^{ij}=0$, which is compatible with quantization, since it can be always gauged away \cite{Kontsevich}, we end up with $\Gamma ^{ij}=-\omega ^{ij}$. 

The equation (\ref{6}) in
the second order is:%
\begin{eqnarray}
2\,\Gamma ^{[ji]k}\pi_k &=& G^{ijk}\pi_k\,, \label{7}\\ G^{ijk}&=&-\frac12\,\omega ^{lk}\partial _{l}\omega^{ij}+\frac14\,\omega^{il}\partial_l\omega^{jk}-\frac14\,\omega^{jl}\partial_l\omega^{ik}\,. \notag
\end{eqnarray}%
The symmetry of $\Gamma^{ijk}$ in the last two indices implies the identity: $\Gamma ^{[ji]k}+\Gamma ^{[kj]i}+\Gamma ^{[ik]j}\equiv0$, which in turn means the consistency condition for the solution of (\ref{7}),
\begin{equation}
G^{ijk}+G^{kij}+G^{jki}=0\,.\label{cycl}
\end{equation}
This condition we call cyclicity and it is satisfied due to the Jacobi identity (\ref{JI}).
A solution of the equation (\ref{7}) is given by:%
\begin{equation}
\Gamma ^{ijk}=-\frac{1}{6}\left(G^{ijk}+G^{ikj}\right)=\frac{1}{24}\,\omega ^{km}\partial _{m}\omega ^{ij}+\frac{1}{24}%
\,\omega ^{jm}\partial _{m}\omega ^{ik}.  \label{8}
\end{equation}%
It is symmetric in $j$ and $k$ by the construction and using (\ref{cycl}) one can make sure that it satisfies (\ref{7}).

Now let us discuss how to construct the solution in the higher orders. It is convenient to
represent the right hand side of (\ref{6}) as
\begin{eqnarray}
&& \omega^{ij}= \omega_n^{ij}+{\cal O}\left(\alpha^{n+1}\right)\,,  \label{9}\\
&&\omega_0^{ij}=\omega ^{ij}(y),\,\,\,\omega_1^{ij}=\omega ^{ij}+\alpha\,\partial _{l}\omega^{ij}\Gamma^{lk}\pi_k.\nonumber
\end{eqnarray}
We also introduce
corresponding notations for $ x^i$,
\begin{equation}
 x^i= x^i_n +{\cal O}\left(\alpha^{n+1}\right)\,,\qquad  x^i_{n+1}= x^i_n
+\alpha^{n+1}\Gamma^{ij_{1}...j_{n+1}}\pi_{j_{1}}...\pi_{j_{n+1}}\,.  \label{10}
\end{equation}
One may easily check that for any analytic function $f(x)$ one has:
\begin{equation}\label{45}
    f\left(x_n\right)=f_n+{\cal O}\left(\alpha^{n+1}\right).
\end{equation}

Suppose, that the expansion (\ref{4}) is known up to the $n$-th
order, i.e., the equation
\begin{equation}
\{ x_n^i, x_n^j\}=\alpha\, \omega_{n-1}^{ij} +{\cal O}\left(\alpha^{n+1}\right)\,,
\label{coman}
\end{equation}
holds true. 
In order to construct the $(n+1)$-th order in the decomposition we have
to solve the next order equation:
\begin{equation}
\{{x}_{n+1}^{i},{x}_{n+1}^{j}\} =\alpha\,{\omega}%
_{n}^{ij}+{\cal O}\left(\alpha^{n+2}\right) ~.  \label{11}
\end{equation}%
Using (\ref{10}) we represent it in the form
\begin{equation}
\alpha ^{n+1}(n+1)\Gamma^{[ji]j_{1}...j_{n}}\pi_{j_{1}}...\pi_{j_{n}} =G_{n+1}^{ij} +{\cal O}\left(\alpha^{n+2}\right)\, ,\label{12a}
\end{equation}
where
\begin{equation}
G_{n+1}^{ij}=\alpha\,{\omega}_{n}^{ij}-\{{x}_{n}^{i},{x}_{n}^{j}\} +{\cal O}\left(\alpha^{n+2}\right)\,.\label{12b}
\end{equation}
The above equation defines $G_{n+1}^{ij}$ up to the terms ${\cal O}\left(\alpha^{n+2}\right)$, and we do not include any higher-order terms in $G_{n+1}^{ij}$. Taking into account (\ref{coman}) one writes,
\begin{equation}
\alpha^{-(n+1)}\ G_{n+1}^{ij}=\frac{1}{n!}\,\left[ \frac{d^{n}}{d\alpha ^{n}}\omega ^{ij}\left(x_n\right) \right] _{\alpha =0}  -\sum_{m=1}^{n}\left\{ \Gamma ^{i\left( n+1-m\right) }\left( \pi
\right) ^{n+1-m},\Gamma ^{j\left( m\right) }\left( \pi \right) ^{m}\right\}\,.\label{13}
\end{equation}
One can also represent it in the form
\begin{equation}
G_{n+1}^{ij}=\alpha^{n+1}\ G^{ijj_{1}\dots j_{n}}(y)\ \pi_{j_{1}}...\pi_{j_{n}}\,\label{13a}
\end{equation}
where the coefficient functions $G^{ijj_{1}\dots j_{n}}$ are antisymmetric in first two indices and symmetric in last $n$ by the construction. 

Thus, (\ref{11}) implies the algebraic equation
\begin{equation}
(n+1)\Gamma^{[ji]j_{1}...j_{n}}=G^{ijj_{1}\dots j_{n}}\,.\label{GGam}
\end{equation}
Like in the case of the eq. (\ref{7}), the symmetry of the tensors $\Gamma^{ijj_{1}...j_{n}}$ in the last $n+1$ indices yields the consistency condition on the right hand side of (\ref{GGam}), namely the cyclicity relation:
\begin{equation}
G^{ijj_{1}\dots j_{n}}+G^{j_{1}ijj_{2}\dots j_{n}}+G^{jj_{1}ij_{2}\dots j_{n}}=0.  \label{cyclG}
\end{equation}
The condition (\ref{cycl}) holds true as a consequence of the Jacobi identity (\ref{JI}). The following Lemma shows that the same is valid for (\ref{cyclG}).

\begin{lemma}
\label{L1} The functions $G^{ijj_{1}\dots j_{n}}$ defined in eq. (\ref{13}-\ref{13a}) satisfy the cyclicity relation (\ref{cyclG}).
\end{lemma}

\begin{proof} 
The Jacobi identity (\ref{JI}) can be written as:
\begin{equation}\label{15}
  \{{x}^{k},{\omega}^{ij}(x)\}_p +\{{x}^{j},%
{\omega}^{ki}(x)\}_p +\{{x}^{i},{\omega}^{jk}(x)%
\}_p =0~,
\end{equation}
Since,  $ x^i= x^i_n +{\cal O}\left(\alpha^{n+1}\right)\,,$ and ${\omega}^{ij}(x_n)={\omega}_{n}^{ij}+{\cal O}\left(\alpha^{n+1}\right)$, the above equation implies
\begin{equation}
\{{x}_{n}^{k},{\omega}_{n}^{ij}\} +\{{x}_{n}^{j},%
{\omega}_{n}^{ki}\} +\{{x}_{n}^{i},{\omega}_{n}^{jk}%
\} ={\cal O}\left(\alpha^{n+1}\right)~,  \label{16}
\end{equation}
or equivalently,
\begin{equation}
\left\{{x}_{n}^{k},\alpha\,{\omega}_{n}^{ij}\right\} +\mbox{cycl.}%
(kij)={\cal O}\left(\alpha^{n+2}\right)~.
\end{equation}%
Now using (\ref{12b}), one gets,
\begin{equation}
\left\{{x}_{n}^{k},G_{n+1}^{ij}+\left\{ {x}_{n}^{i},{x}_{n}^{j}%
\right\} \right\} +\mbox{cycl.}(kij)={\cal O}\left(\alpha^{n+2}\right) ~.
\label{17}
\end{equation}
Since, the Jacobi identity,
\begin{equation}
\left\{{x}_{n}^{k},\left\{{x}_{n}^{i},{x}_{n}^{j}\right\} \right\}
+\mbox{cycl.}(kij)=0
\end{equation}
holds true at all orders of $\alpha $, including the order $\alpha ^{n+1}$,
from (\ref{17}) one obtains
\begin{equation}
\left\{ x_{n}^{k},G_{n+1}^{ij}\right\} +\mbox{cycl.}(kij)={\cal O}\left(\alpha^{n+2}\right)~,  \label{18}
\end{equation}
meaning that
\begin{equation}
\left\{ y^{k},G^{ijj_{1}\dots j_{n}}(y)\pi_{j_{1}}...\pi_{j_{n}}\right\} +\mbox{cycl.}(kij)=0~.
\end{equation}
Next one calculates
the Poisson bracket, and uses the symmetry of $G^{ijj_{1}\dots j_{n}}$ in the last $n$ indices
to prove the condition (\ref{cyclG}). Also the symmetry of $G^{ijj_{1}\dots j_{n}}$ in the last $n$
indices, implies the cyclic conditions holds for permuations of $(i,j,i_k)$ for any $%
k=1,\dots,n$.
\end{proof}

As long as the consistency condition (\ref{cyclG}) is satisfied, the solution of the equation (\ref{GGam}) is provided by the following 

\begin{lemma}
\label{L2} The tensors
\begin{equation}
\Gamma^{ij_1\dots j_{n+1}}=-\frac {1}{(n+1)(n+2)} \left(
G^{ij_1j_2\dots j_{n+1}}+ G^{ij_2j_1j_3\dots j_{n+1}}+\dots
+G^{ij_{n+1}j_1j_2\dots j_{n}} \right)  \label{GamG}
\end{equation}
are symmetric in the last $n+1$ indices and satisfy the equation (\ref{GGam}).
\end{lemma}

\begin{proof}
The symmetry follows by the construction. Now let us consider the right hand side of (\ref{GamG}) and calculate
\begin{equation*}
  \Gamma^{[j_1i]j_{2}...j_{n+1}}=\Gamma^{j_1ij_{2}...j_{n+1}}-\Gamma^{ij_1j_{2}...j_{n+1}}.
\end{equation*}
 Since $G_{n+1}^{ij_{1}\dots j_{n+1}}$ is antisymmetric in the first two indices one has
\begin{equation*}
 G_{n+1}^{j_1ij_2\dots j_{n+1}}-G_{n+1}^{ij_1j_2\dots j_{n+1}}=- 2
G_{n+1}^{ij_1j_2\dots j_{n+1}}.
\end{equation*}
The cyclic condition (\ref{cyclG}) in $i,j_1,j_2$ implies:
\begin{equation*}
G_{n+1}^{j_1j_2i\dots j_{n+1}}-G_{n+1}^{ij_2j_1\dots j_{n+1}}= -G_{n+1}^{ij_1j_2\dots
j_{n+1}}.
\end{equation*}
Remaining $n-1$ combinations are
treated similarly using (\ref{cyclG}), and the assertion follows immediately.
\end{proof}

This Lemma implies that the tensors (\ref{GamG}) are indeed the coefficient
functions of the expansion (\ref{4}) and the functions $ x^i(y,\pi)$ satisfy the algebra of Poisson brackets (\ref{1}).

For the double coordinates $\tilde x_i$ we may write
\begin{equation}
\tilde x_i=\pi_{i}-\alpha\, j_{i}(y,\pi), \label{19}
\end{equation}
where $j_{i}(y,\pi,\alpha)$ is an arbitrary differentiable function. Calculating the Poisson brackets between coordinates (\ref{4}) and double coordinates (\ref{19}) one finds functions $\delta^i_j(y,\pi)$ and $\varpi_{ij}(y,\pi)$ which together with $\omega^{ij}(x)$ determine the symplectic structure (\ref{3}). Using the inverse expressions $y^i=y^i(x,\tilde x)$ and $\pi_i=\pi_i(x,\tilde x)$ we end up with
\begin{align}
& \delta^{ij}\left( x,\tilde x\right) =\delta^{ij}+\alpha\left(\tilde\partial_{i}\,j_{j}-%
\frac12\,\partial_{j}\omega^{li}\tilde x_{l}\right) +{\cal O}\left(\alpha^{2}\right)\,, \label{20} \\
& \varpi^{ij}\left( x,\tilde x\right) =\alpha\left(
\partial_{j}j_{i}-\partial_{i}j_{j}\right) +{\cal O}\left(\alpha^{2}\right)\,, \notag
\end{align}
where $\tilde\partial_{i}=\partial/\partial \tilde x_i$. For simplicity in what follows we just set $\tilde x_i=\pi_{i}$.

\section{Quasi-Poisson structure and symplectic embeddings}

In this section we discuss the problem of a symplectic realization of a quasi-Poisson manifold. Consider the bracket
\begin{equation}\label{21}
   \{f, g\}_Q=\alpha\, \Theta^{ij}(x)\ \partial_if\partial_jg\,,
\end{equation}
defined by a given bi-vector $\Theta=\frac{1}{2}\Theta^{ij}(x)\partial_i\wedge\partial_j$. The bracket (\ref{21}) is bilinear and antisymmetric, but in general the Jacobi identity can be violated, meaning that the three-bracket,
\begin{equation}\label{3bracket}
  \left\{f,g,h\right\}_Q:= \mbox{$\frac{1}3$} \, \big(\{f,\{g,h\}_Q\}_Q+\{h,\{f,g\}_Q\}_Q+\{g,\{h,f\}_Q\}_Q\big)\,,
\end{equation}
can be different from zero for three arbitrary functions $f$, $g$ and $h$. By the definition the above bracket is also antisymmetric and tri-linear. One may also write,
\begin{eqnarray}
 \left\{f,g,h\right\}_Q=\alpha^2\, \Pi^{ijk}\ \partial_if \partial_jg \partial_kh.
\end{eqnarray}
Using the definition and the properties of the brackets (\ref{21}) and (\ref{3bracket}) one may check that the following combination of the two and three brackets is identically zero\footnote{In fact, this expression relates two different ways of rebracketing the expression $\{f,\{g,\{h,k\}\}\}$, see \cite{MSS1} for details.}:
\begin{eqnarray}\label{HJI}
&& - \{f,g, \{h,k\}_Q\}_Q+ \{g,h ,\{k,f\}_Q\}_Q- \{h,k, \{f,g\}_Q\}_Q \\
&& +  \{k,f ,\{g,h\}_Q\}_Q- \{h,f, \{g,k\}_Q\}_Q+ \{g,k, \{f,h\}_Q\}_Q\notag \\
&& + \{f, \{g,h,k\}_Q\}_Q-\{k, \{f,g,h\}_Q\}_Q+\{h, \{k,f,g\}_Q\}_Q-\{g, \{h,k,f\}_Q\}_Q\equiv0.\notag
\end{eqnarray}
Since (\ref{HJI}) should be valid for any functions $f$, $g$, $h$ and $k$ one obtains the following identity involving $\Theta$ and $\Pi$:
\begin{eqnarray}\label{HJI1}
&& \Pi^{ijm}\partial_m\Theta^{kl}-\Pi^{jkm}\partial_m\Theta^{li}+\Pi^{klm}\partial_m\Theta^{ij} \\ &&-\Pi^{lim}\partial_m\Theta^{jk}-\Pi^{ikm}\partial_m\Theta^{jl}+\Pi^{jlm}\partial_m\Theta^{ki}\notag \\
&&  +\Theta^{lm}\partial_m\Pi^{ijk}-\Theta^{im}\partial_m\Pi^{jkl}+\Theta^{jm}\partial_m\Pi^{kli}-\Theta^{km}\partial_m\Pi^{lij}\equiv0.\notag
\end{eqnarray}

In this section we will consider the following problem, how to construct a $2N$-dimensional symplectic manifold $\mathcal{S}$ with coordinates
 $\left( x^{i},\tilde x_{i}\right) ,$ and Poisson brackets
\begin{eqnarray}
\{x^i, x^j\}_p&=&\alpha\, \omega^{ij}(x,\tilde x),  \label{23} \\
\{x^i, \tilde x_j\}_p&=&\delta^i_j(x,\tilde x)=\delta^i_j+\alpha\,{\delta^{(1)}}^i_j(x,\tilde x)+{\cal O}(\alpha^2),\nonumber\\
\{\tilde x_i, \tilde x_j\}_p&=&\alpha\,\varpi_{ij}(x,\tilde x),\nonumber
\end{eqnarray}
such that
\begin{itemize}
\item a restriction of (\ref{23}) on the subspace generated by $x^i$ would reproduce the original quasi-Poisson structure (\ref{21}), i.e., 
$\{f(x), g(x)\}_p|_{\tilde x=0}=\{f, g\}_Q,$ meaning that
\begin{equation}\label{24}
  \omega^{ij}(x,\tilde x)=\sum_{n=0}^{\infty}\Theta^{ij(n)}(x)(\alpha \tilde x)^{n}=\Theta^{ij}(x)+\alpha\,\Theta^{ijk}(x)\,\tilde x_k +\dots\,.
\end{equation}
\item For Poisson bi-vector $\Theta$ the algebra (\ref{23}) should restore the symplectic realization of a Poisson structure previously defined in (\ref{2}). That is, the tensors $\Theta^{ij(n)}(x),$ $n\geq1$, should be proportional to $\Pi^{ijk}$ and its derivatives.
\end{itemize}
By the definition $\Theta^{ij(n)}(x)$ are antisymmetric in first two indices and symmetric in last $n$ indices. There is no apriori symmetry beween the first pair and the rest of the indices, but we impose
the requirement that $\Theta^{ijk_1k_2\dots k_n}$ transforms under the permutations according
to the following Young tableau.
\begin{equation}
\ytableausetup{mathmode}
\begin{ytableau}
k_1 & k_2 & \none[\dots] & k_n \\
i \\
j
\end{ytableau} \label{YOmega}
\end{equation}

As in the previous section, we are looking for the perturbative expansion of the coordinates $x^i$ and momenta $p_i$ in terms of the Darboux variables $\left( y^{i},\pi _{i}\right) $, satisfying the canonical PB (\ref{3}).
The generalized Bopp shift is given by (\ref{4}), and $\tilde x_i=\pi_i$. It is convenient to introduce here the following notations
\begin{eqnarray}\label{39}
\omega^{ij}(x,\tilde x) & = & \tilde\omega^{ij}_n(x,\tilde x)+{\cal O}\left(\alpha^{n+1}\right)\,,\\ \tilde\omega^{ij}_{n+1}(x,\tilde x)&=&\tilde\omega^{ij}_n(x,\tilde x)+\Theta^{ij(n+1)}(x)\,( \alpha\tilde x)^{n+1}.\nonumber
\end{eqnarray}
The expression for $x^i_n$ as a truncation of a generalized Bopp shift is defined by (\ref{10}). Here we introduce also
 \begin{equation}\label{44}
\omega^{ij}\left(y^{i}+\sum_{n=1}^{\infty}\Gamma^{i\left( n\right) }\left( y\right)
\left( \alpha\pi\right) ^{n},\pi\right)=    \omega^{ij}_n+{\cal O}\left(\alpha^{n+1}\right)\,.
\end{equation}
The difference between $\tilde\omega^{ij}_n$ and $\omega^{ij}_n$ is that the first one is a truncation of the series (\ref{24}) written in terms of the original phase space coordinates $x$ and $\tilde x$, while the second is a truncation of the corresponding series expressed in Darboux coordinates $y$ and $\pi$.
We stress also that the structure of the function $\omega^{ij}_n$ in (\ref{44}) is different from the corresponding expression (\ref{9}) in the Poisson case, because of the presence of the terms $\Theta^{ij(n)}(x)(\alpha \tilde x)^{n}$ in the expression for $\omega^{ij}(x,\tilde x)$. In particular,
\begin{equation*}
  \omega_0^{ij}=\Theta ^{ij}(y),\,\,\,\omega_1^{ij}=\Theta ^{ij}+\alpha\left(\partial _{l}\Theta^{ij}\Gamma^{lk}\pi_k+\Theta^{ijk}\pi_k\right)\, ,
\end{equation*}
etc. Taking into account (\ref{45}) and (\ref{39}) one may see that
\begin{eqnarray}
  && \omega^{ij}_n=\tilde\omega^{ij}_n\left(x_n,\pi\right)+{\cal O}\left(\alpha^{n+1}\right)= \label{46}\\
  &&\Theta^{ij}(x_n)+\alpha\Theta^{ijk}(x_{n-1})\pi_k +\dots+\Theta^{ij(n-1)}(x_1)\,(\alpha\pi)^{n-1}+\Theta^{ij(n)}(y)\,(\alpha\pi)^n+{\cal O}\left(\alpha^{n+1}\right)\,.\notag
\end{eqnarray}
Also it is useful to write
\begin{equation}\label{50}
    \omega^{ij}_n=\tilde\omega^{ij}_{n-1}\left(x_n,\pi\right)+\Theta^{ij(n)}(y)\,(\alpha \pi)^{n}+{\cal O}\left(\alpha^{n+1}\right)\,.
\end{equation}

To define the coefficient functions  $\Gamma^{i\left( n\right) }$ and $\Theta^{ij(n)}$ of the generalized Bopp shift (\ref{4}) and the series (\ref{24}) correspondingly, one should solve the equation
\begin{align}
& \left\{ y^{i}+\sum_{n=1}^{\infty }\Gamma ^{i\left( n\right) }\left(
y\right) \left( \alpha \pi \right) ^{n},y^{j}+\sum_{n=1}^{\infty }\Gamma
^{j\left( n\right) }\left( y\right) \left( \alpha \pi \right) ^{n}\right\} =
\label{26} \\
& \alpha \sum_{n=0}^{\infty}\Theta^{ij(n)}\left( y^{i}+\sum_{m=1}^{\infty }\Gamma ^{i\left(
m\right) }\left( y\right) \left( \alpha \pi \right) ^{m}\right)(\alpha\pi)^{n}.  \notag
\end{align}%
Like in the previous section in the first order in $\alpha$ one obtains:%
\begin{equation*}
\Gamma ^{[ji]}=\Theta ^{ij},
\end{equation*}%
with a solution $\Gamma ^{ij}=-\Theta ^{ij}/2$. The second order in $\alpha$ gives:%
\begin{equation}
2\Gamma ^{[ji]k} =-\frac{1}{2}\Theta ^{lk}\partial _{l}\Theta^{ij}+\frac{1}{4}\Theta^{il}\partial_l\Theta^{jk}-\frac{1}{4}\Theta^{jl}\partial_l\Theta^{ik}+\Theta^{ijk}.  \label{27}
\end{equation}%
The consistency condition for the equation (\ref{27}), implies the relation
\begin{equation}\label{28}
  3\, \Pi^{ijk}+\Theta^{ijk}+\Theta^{kij}+\Theta^{jki}=0\,,
\end{equation}
which in turn can be interpreted as an equation on the coefficient function $\Theta^{ijk}$. A solution is:
\begin{equation}\label{29}
  \Theta^{ijk}=-\Pi^{ijk}\,.
\end{equation}
Since, $\Theta^{ijk}$ is antisymmetric in all indices it will not contribute to a solution of the equation (\ref{27}) such that:%
\begin{equation}
\Gamma ^{ijk}=\frac{1}{24}\Theta ^{km}\partial _{m}\Theta ^{ij}+\frac{1}{24}%
\Theta ^{jm}\partial _{m}\Theta ^{ik}\,.  \label{30}
\end{equation}%
The crucial difference of the current situation with respect to the Poisson case is that the consistency condition (\ref{28}) of the algebraic equation (\ref{27}), is not satisfied automatically, but instead provides the equation for the definition of the corrections $\Theta^{ij(n)}(x)(\alpha \tilde x)^{n}$, $n\geq1$, to the given bi-vector $\Theta^{ij}(x)$ in (\ref{24}), which are needed for the Jacobi identity for the algebra (\ref{23}) to hold.

Following the logic of the previous section suppose we know the solution of the equation (\ref{26}) up to the $n$-th order in $\alpha$. It means that the expressions for $x_n^i$ and $\omega^{ij}_{n-1}$ are known such that the equation 
\begin{equation}\label{48}
    \{x_n^i, x_n^j\}=\alpha\, \omega^{ij}_{n-1}+{\cal O}\left(\alpha^{n+1}\right)\,,
\end{equation}
holds true. Taking into account (\ref{46}) the consistency condition for the above equation can be written as
\begin{equation}\label{47}
   \{x_{n-1}^{k},{\tilde\omega}^{ij}_{n-1}(x_{n-1},\pi)\} +\{x_{n-1}^{j},%
{\tilde\omega}^{ki}_{n-1}(x_{n-1},\pi)\} +\{x_{n-1}^{i},{\tilde\omega}^{jk}_{n-1}(x_{n-1},\pi)%
\} ={\cal O}\left(\alpha^{n}\right)~.
\end{equation}
To obtain the next order contributions in the expansions  (\ref{4}) and (\ref{24}) one needs to solve the equation
\begin{equation}\label{49}
    \{x_{n+1}^i, x_{n+1}^j\}=\alpha\, \omega^{ij}_{n}+{\cal O}\left(\alpha^{n+2}\right)\,.
\end{equation}
Like in the Poisson case discussed in the previous section, the above equation is equivalent to the equations (\ref{12a}-\ref{12b}), with the difference that now $\omega^{ij}_{n}$ contains the corrections $\Theta^{ij(n)}\,(\alpha \tilde x)^{n}$. In particular
\begin{equation}\label{50}
    \omega^{ij}_n=\tilde\omega^{ij}_{n-1}\left(x_n,\pi\right)+\Theta^{ij(n)}(y)\,(\alpha \pi)^{n}+{\cal O}\left(\alpha^{n+2}\right)\,,
\end{equation}
where the function $\Theta^{ij(n)}(y)$  is yet unknown and should be found solving the consistency condition  for the equation (\ref{49}):
\begin{equation}\label{51}
   \{{x}_n^{k},{\omega}^{ij}_{n}\} +\{{x}_n^{j},%
{\omega}^{ki}_{n}\} +\{{x}_n^{i},{\omega}^{jk}_{n}%
\} ={\cal O}\left(\alpha^{n+1}\right)~.
\end{equation}
Taking into account (\ref{50}) the above equation can be written as
\begin{eqnarray}
  && \{{x}_n^{k},\tilde\omega^{ij}_{n-1}\left(x_n,\pi\right)\} +\{{x}_n^{j},%
\tilde\omega^{ki}_{n-1}\left(x_n,\pi\right)\} +\{{x}_n^{i},\tilde\omega^{jk}_{n-1}\left(x_n,\pi\right)%
\}\label{52}\\&&+n\,\alpha^n\,\left(\Theta^{ijk(n-1)}+\Theta^{kij(n-1)}+\Theta^{jki(n-1)}\right)( \pi)^{n-1} ={\cal O}\left(\alpha^{n+1}\right)~.\nonumber
\end{eqnarray}
That is, the coefficient function $\Theta^{ij(n)}$ should satisfy the equation
\begin{equation}\label{53}
    n\,\alpha^n\,\left(\Theta^{ijk(n-1)}+\Theta^{kij(n-1)}+\Theta^{jki(n-1)}\right)( \pi)^{n-1}+F^{ijk}_{n}=0\,,
\end{equation}
where
\begin{equation}\label{54}
    F^{ijk}_{n}=\{{x}_n^{k},\tilde\omega^{ij}_{n-1}\left(x_n,\pi\right)\} +\{{x}_n^{j},%
\tilde\omega^{ki}_{n-1}\left(x_n,\pi\right)\} +\{{x}_n^{i},\tilde\omega^{jk}_{n-1}\left(x_n,\pi\right)%
\}+{\cal O}\left(\alpha^{n+1}\right)\,.
\end{equation}
In particular, $F^{ijk}_{1}=3\,\Pi^{ijk}$, etc. The above definition and the relation (\ref{47}) imply that
\begin{equation}\label{55}
    F^{ijk}_{n}=\alpha^n\,F^{ijk(n-1)}(y)\,( \pi)^{n-1}\,.
\end{equation}
The coefficient functions $F^{ijk(n-1)}$ by the definition (\ref{54}) are antisymmetric in the first three indices $ijk$ and symmetric in last $n-1$ indices.

\subsection{Next to the leading order}

The solution of the equation (\ref{53}) in the first order in $\alpha$ is given by (\ref{29}). However, just like it happened with the equation (\ref{GGam}), its non-trivial nature and the corresponding consistency condition can be seen in the next to the leading order. That is why in this subsection we discuss in details the equation (\ref{54}) in the second order in $\alpha$, where it is equivalent to the  algebraic equation on the coefficient function $\Theta^{ijkl}$:
\begin{equation}\label{32}
 2( \Theta^{ijkl}+\Theta^{kijl}+\Theta^{jkil})-F^{ijkl}=0\,,
\end{equation}
with
\begin{eqnarray}
\label{33}
  &&-F^{ijkl}=\Theta^{kml}\partial_m\Theta^{ij}+ \Theta^{jml}\partial_m\Theta^{ki}+\Theta^{iml}\partial_m\Theta^{jk}\\
  &&\Theta^{km}\partial_m\Theta^{ijl}+\Theta^{jm}\partial_m\Theta^{kil}+\Theta^{im}\partial_m\Theta^{jkl}\notag \\
  &&\frac{1}{2}\Theta^{ijm}\partial_m\Theta^{kl}+\frac{1}{2}\Theta^{kim}\partial_m\Theta^{jl}+\frac{1}{2}\Theta^{jkm}\partial_m\Theta^{il}\,.\notag
\end{eqnarray}
Since $\Theta^{ijkl}$ should be antisymmetric in first two indices and symmetric in last two indices one obtains the consistency condition for the solution of the algebraic equation (\ref{32}):
\begin{equation}\label{34}
  F^{ijkl}-F^{lijk}+F^{klij}-F^{jkli}=0.
\end{equation}
Taking into account (\ref{29}) one finds that this equation is exactly the relation (\ref{HJI1}).

The following combination,
\begin{equation}\label{35}
 \Theta^{ijkl}= \frac{1}{8}\left(F^{ijkl}+F^{ijlk}\right)\,,
\end{equation}
is symmetric in $kl$ and antisymmetric in $ij$ by the construction, as well as satisfies the equation (\ref{32}) due to (\ref{34}). One writes:
\begin{eqnarray}
  &&\Theta^{ijkl}=\frac{3}{16}\Pi^{jlm}\partial_m\Theta^{ki}+\frac{3}{16}\Pi^{jkm}\partial_m\Theta^{li}-
  \frac{3}{16}\Pi^{ilm}\partial_m\Theta^{kj}-\frac{3}{16}\Pi^{ikm}\partial_m\Theta^{lj} \label{35a}\\
  &&-\frac{1}{8}\Theta^{km}\partial_m\Pi^{ijl}-\frac{1}{8}\Theta^{lm}\partial_m\Pi^{ijk}\,.\notag
\end{eqnarray}

Now one may use the Lemma \ref{L2} to construct the third order expression for $x^i$. From the equation
\begin{equation}\label{35b}
    \{x_3^i, x_3^j\}=\alpha \omega^{ij}_2+{\cal O}\left(\alpha^{4}\right)\,,
\end{equation}
where
\begin{equation*}
    \omega_2^{ij}=\Theta^{ij}(x_2)+\alpha\Theta^{ijk}(x_1)\pi_k +\alpha^2\Theta^{ijkl}(y)\pi_k \pi_l+{\cal O}\left(\alpha^{3}\right)\,,
\end{equation*}
we find that
\begin{eqnarray}
  &&G^{ijmn}=\partial_l\Theta^{ij}\Gamma^{lmn}+\frac{1}{8}\partial_l\partial_{k}\Theta^{ij}\Theta^{lm}\Theta^{kn}\label{G3}\\
  &&+\frac{1}{4}\partial_l\Pi^{ijm}\Theta^{ln}+\frac{1}{4}\partial_l\Pi^{ijn}\Theta^{lm}+\Theta^{ijmn} \notag\\
  &&-\frac{1}{2}\Theta^{jl}\partial_l\Gamma^{imn}+\frac{1}{2}\Theta^{il}\partial_l\Gamma^{jmn}-\frac{1}{2}\Gamma^{jml}\partial_l\Theta^{in}-\frac{1}{2}\Gamma^{jnl}\partial_l\Theta^{im} \notag\\
  &&+\frac{1}{2}\Gamma^{iml}\partial_l\Theta^{jn}+\frac{1}{2}\Gamma^{inl}\partial_l\Theta^{jm}\,.\notag
\end{eqnarray}
After simplifications one obtains the expression
\begin{eqnarray}
  &&\Gamma^{ijmn}=-\frac{1}{12}\left(G^{ijmn}+G^{imjn}+G^{injm}\right)\label{Gamma3}\\
  &&=-\frac{1}{12}\left(2\Gamma^{lmn}\partial_l\Theta^{ij}+2\Gamma^{ljn}\partial_l\Theta^{im}+2\Gamma^{ljm}\partial_l\Theta^{in}\right.\notag\\
 &&\left. +\frac{1}{12}\Theta^{lm}\Theta^{kn}\partial_l\partial_k\Theta^{ij}+\frac{1}{12}\Theta^{lj}\Theta^{kn}\partial_l\partial_k\Theta^{im}+\frac{1}{12}\Theta^{lj}\Theta^{km}\partial_l\partial_k\Theta^{in}\right)\,.\notag
\end{eqnarray}
That is,
\begin{equation}\label{x3}
    x^i_3=y^i-\frac{\alpha}{2}\Theta^{ij}\pi_j+\frac{\alpha^2}{12}\Theta^{km}\partial_m\Theta^{ij}\pi_j\pi_k-\frac{\alpha^3}{48}\Theta^{km}\Theta^{ln}\partial_m\partial_n\Theta^{ij}\pi_j\pi_k\pi_l\,.
\end{equation}

Note that the expression for the generalized Bopp shift is absolutely identical to the Poisson case. The reason for that is the fact that although the tensors $\Theta^{ij(n)}$, for $n\geq1$, enter the right hand side of the equation (\ref{49}), i.e., in the definition of the tensors $G^{ij(n)}$, they are constructed as a linear combination of the tensors $F^{ijk(n-1)}$, which are antisymmetric in the first three indices. According to the Lamma \ref{L2}, the solution of the equation (\ref{49}), i.e., the expression for $x^i_{n+1}$ is constructed by the symmetrization of the last $(n+1)$ indices of the tensors $G^{ij(n)}$, which annihilates any contribution from the tensors $F^{ijk(n-1)}$, and consequently from $\Theta^{ij(n)}$, for $n\geq1$.

\subsection{Higher order terms $\Theta^{ij(n)}$, $n\geq2$.}

Once we know the structure of the consistency condition for the equation (\ref{53}), we may proceed to the higher orders and prove the following

\begin{lemma}
\label{L3} The tensors $F^{ijkl(n-2)}$, defined in (\ref{55}) satisfy the relation:
\begin{equation}\label{56}
    F^{ijkl(n-2)}-F^{l ijk(n-2)}+F^{kl ij(n-2)}-F^{jkli(n-2)}=0\,.
\end{equation}
\end{lemma}

\begin{proof}
To prove the above statement first we observe that the relation (\ref{56}) is equivalent to the condition
\begin{equation}\label{57}
   \{x^l_n, F^{ijk}_{n}\}- \{x^k_n, F^{lij}_{n}\}+ \{x^j_n, F^{kli}_{n}\}- \{x^i_n, F^{jkl}_{n}\}={\cal O}\left(\alpha^{n+1}\right)\,.
\end{equation}
Using the definition (\ref{54}) we may write the left-hand side of this equation as:
\begin{align}
&\{x^l_n, \{{x}_n^{i},\tilde\omega^{jk}_{n-1}\left(x_n,\pi\right)%
\}\}+\{x^l_n, \{{x}_n^{k},\tilde\omega^{ij}_{n-1}\left(x_n,\pi\right)\}\} +\{x^l_n,\{{x}_n^{j},%
\tilde\omega^{ki}_{n-1}\left(x_n,\pi\right)\}\} \label{58} & \\
- &\{x^k_n,\{{x}_n^{l}, \tilde\omega^{ij}_{n-1}\left(x_n,\pi\right)\}\} -\{x^k_n, \{{x}_n^{j},%
\tilde\omega^{li}_{n-1}\left(x_n,\pi\right)\}\} -\{x^k_n, \{{x}_n^{i},\tilde\omega^{jl}_{n-1}\left(x_n,\pi\right)%
\}\} \notag&  \\
+ &\{x^j_n,\{{x}_n^{k},\tilde\omega^{li}_{n-1}\left(x_n,\pi\right)\}\} +\{x^j_n,\{{x}_n^{i},%
\tilde\omega^{kl}_{n-1}\left(x_n,\pi\right)\}\} +\{x^j_n, \{{x}_n^{l},\tilde\omega^{ik}_{n-1}\left(x_n,\pi\right)%
\}\} \notag & \\
-& \{x^i_n,\{{x}_n^{j}, \tilde\omega^{kl}_{n-1}\left(x_n,\pi\right)\}\} -\{x^i_n, \{{x}_n^{l},%
\tilde\omega^{jk}_{n-1}\left(x_n,\pi\right)\}\}-\{x^i_n, \{{x}_n^{k},\tilde\omega^{lj}_{n-1}\left(x_n,\pi\right)%
\}\}+{\cal O}\left(\alpha^{n+1}\right)\,.&\notag
\end{align}
The Jacoby identity for the functions $x^l_n$, $x^k_n$ and $\tilde\omega^{ij}_{n-1}\left(x_n,\pi\right)$ reads
\begin{equation*}\label{59}
    \{x^l_n, \{{x}_n^{k},\tilde\omega^{ij}_{n-1}\left(x_n,\pi\right)\}\}+\{\tilde\omega^{ij}_{n-1}\left(x_n,\pi\right),\{x^l_n, {x}_n^{k},\}\}+\{x^k_n, \{\tilde\omega^{ij}_{n-1}\left(x_n,\pi\right),{x}_n^{l}\}\}\equiv0\,,
\end{equation*}
and holds in all orders in $\alpha$.
Since by (\ref{48}) the functions $x^l_n$ satisfy the equation,  $\{x_n^l, x_n^k\}=\alpha\, \omega^{lk}_{n-1}+{\cal O}\left(\alpha^{n+1}\right)$, from the above identity we conclude that
\begin{equation*}\label{60}
    \{x^l_n, \{{x}_n^{k},\tilde\omega^{ij}_{n-1}\left(x_n,\pi\right)\}\}-\{x^k_n, \{{x}_n^{l},\tilde\omega^{ij}_{n-1}\left(x_n,\pi\right)\}\}=\{\alpha\omega^{lk}_{n-1},\tilde\omega^{ij}_{n-1}\left(x_n,\pi\right)\}+{\cal O}\left(\alpha^{n+1}\right)\,.
\end{equation*}
The latter means that the expression (\ref{58}) can be rewritten as
\begin{align}
&\{\alpha\omega^{li}_{n-1},\tilde\omega^{jk}_{n-1}\left(x_n,\pi\right)\}-\{\alpha\omega^{kj}_{n-1},\tilde\omega^{li}_{n-1}\left(x_n,\pi\right)\}  & \label{61} \\
  +&\{\alpha\omega^{lk}_{n-1},\tilde\omega^{ij}_{n-1}\left(x_n,\pi\right)\}-\{\alpha\omega^{ji}_{n-1},\tilde\omega^{lk}_{n-1}\left(x_n,\pi\right)\}  &\notag \\
+ &\{\alpha\omega^{lj}_{n-1},\tilde\omega^{ki}_{n-1}\left(x_n,\pi\right)\}-\{\alpha\omega^{ik}_{n-1},\tilde\omega^{lj}_{n-1}\left(x_n,\pi\right)\}  +{\cal O}\left(\alpha^{n+1}\right)\,.&\notag
\end{align}
And finally, we observe that by (\ref{46}), $\alpha\tilde\omega^{jk}_{n-1}\left(x_n,\pi\right)=\alpha\omega^{jk}_{n-1}+{\cal O}\left(\alpha^{n+1}\right)$, such that 
\begin{eqnarray*}
&&\{\alpha\omega^{li}_{n-1},\tilde\omega^{jk}_{n-1}\left(x_n,\pi\right)\}-\{\alpha\omega^{kj}_{n-1},\tilde\omega^{li}_{n-1}\left(x_n,\pi\right)\}   = \label{62} \\
  &&\alpha\{\omega^{li}_{n-1},\omega^{jk}_{n-1}\}-\alpha\{\omega^{kj}_{n-1},\omega^{li}_{n-1}\} +{\cal O}\left(\alpha^{n+1}\right)={\cal O}\left(\alpha^{n+1}\right)\,.\notag
\end{eqnarray*}
That is, (\ref{61}) is ${\cal O}\left(\alpha^{n+1}\right)$, and the condition (\ref{57}) holds true.
\end{proof}

Since the consistency condition (\ref{56}) for the equation (\ref{53}) is satisfied due to the above Lemma, the solution can be constructed according to the logic of the Lemma \ref{L2}, i.e., taking the symmetrization of the tensors $F^{ijk(n-1)}$ in the last $n$ indices. The corresponding expression is given by
\begin{equation}\label{63}
  \Theta^{ij(n)}=-\frac{1}{n(n+2)}\left(F^{ijl_1l_2...l_n}+F^{ijl_2 l_1...l_n}+...+F^{ijl_n l_1...l_{n-1}}\right).
\end{equation}

In the conclusion of this section we would like to mention some possible applications of the proposed construction. It is remarkable that if in the L$_\infty$ bootstrap programe \cite{BBKL} one selects the quasi-Poisson bracket (\ref{21}) as the initial setup i.e., $\ell_2(f,g)=\{f,g\}_Q$,  for $f,g\in X_0$, and $\ell_1(f)=\partial_if\in X_{-1}$, then
the structure of the coefficient functions (\ref{63}) are absolutely the same as the ones which determines the brackets, $\ell_{n+2}(f,g,A^n)=\Theta^{ij(n)}\,\partial_i\,f\partial_j\,g\,A^n\in X_0$, with $A\in X_{-1}$.  At that the products $\ell_{n+1}(f,A^n)\in X_{-1}$ carry one vector index and that is why are different from the expression $\Gamma^{i(n)}\,\partial_i\,f\,A^n$. However, they can be constructed using the similar logic of the Lemmas \ref{L1} and \ref{L2}. It would be interesting to understand the precise relation between the symplectic realizations and the L$_\infty$ construction proposed in \cite{BBKL}.

Another interesting application is the non-associative deformation quantization. In the same way like like the symplectic realization of the Poisson manifolds were used for the construction of the associative star products representing the quantization of a given Poisson structure \cite{Hammou:2001cc,GraciaBondia:2001ct,Pachol:2015qia}, the obtained in this Section symplectic realizations of the quasi-Poisson structures can be useful as a starting point for the construction of the non-associative star product compatible with topological limit \cite{CoSch,MK1,Herbst}.

\section{Examples}

\subsection{Constant R-flux algebra}

As a first exemple we consider the phase space algebra which appeared in the context of closed string theory in a presence of non-geometric constant $R$-flux \cite{BP,Lust:2010iy,Blumenhagen:2011ph}. This is possibly the most simple physicallly motivated exemple of the quasi-Poisson structure which reads:
\begin{eqnarray}
\{x^I, x^J\}_R=\Theta^{IJ}(x)=\begin{pmatrix}
  \frac{\ell_s^3}{\hbar^2}\, R^{ijk}\, p_k &  \delta^i_j   \\
 - \delta^i_j  & 0 \end{pmatrix}
\qquad \mbox{with}\qquad  x=(x^I) =({\bf x},{\bf p}).\label{rb}
\end{eqnarray}
where $R^{ijk}=R\, \varepsilon^{ijk}$, and $\varepsilon^{123}=+1$. Note that making, $p\to x$ and $x\to-p$, one obtains the monopole algebra corresponding to a constant magnetic charge distribution $\rho(x)=R$, \cite{BaLu}.  

In this case the jacobiator,
\begin{eqnarray}
\Pi^{IJK}=\begin{pmatrix}
  \frac{\ell_s^3}{\hbar^2}\, R^{ijk} &  0   \\
  0 & 0 \end{pmatrix}\,,\label{pirb}
\end{eqnarray}
does not vanish, but is constant. Consequently the only non-vanishing $\Theta^{IJ(n)}$ with $n\geq1$ is $\Theta^{IJK}=-\Pi^{IJK}$. The corresponding Bopp shift reads
\begin{equation}\label{op2}
 x^I=y^I-\mbox{$\frac12$}\, \Theta^{IJ}(y)\, \pi_J\ .
\end{equation}
And finally for the symplectic realization one finds,
\begin{eqnarray}
\{x^I,x^J\}_p &=&  \Theta^{IJ}(x)- \Pi^{IJK}\, \tilde
x_K \ , \notag \\[4pt]
\{x^I,\tilde x_J\}_p &=& \delta^I{}_J + \mbox{$\frac12$}\,
(\partial_J\Theta^{IK})\, \tilde x_K \ , \notag \\[4pt]
\{\tilde x_I,\tilde x_J\}_p &=& 0 \ .
\end{eqnarray}
The letter is also convenient to write in components. The non-vanishing Poisson brackets between the extended phase space coordinates $({\bf x},{\bf p},{\bf\tilde x},{\bf\tilde p})$ are given by
\begin{eqnarray}
\{x^i,p_j\}&=&\{x^i,\tilde x_j\}\ \ =\ \ -\{
\tilde p^i,p_j\}\ \ =\ \ \delta^i{}_j \ , \nn\\[4pt]
\{x^i,x^j\}&=& \mbox{$ \frac{\ell_s^3}{\hbar^2}$}\, R^{ijk}\, \big(p_k-\tilde x_k\big)\ , \nn\\[4pt]	
\{x^i,\tilde p^j\}&=&\{\tilde p^i,x^j\} \ \ = \ \
-\mbox{$ \frac{\ell_s^3}{2\hbar^2}$}\,  R^{ijk}\, \tilde x_k \ .\label{PB1}
\end{eqnarray}

\subsection{Quasi-Poisson structure isomorphic to the Malcev algebra of octonions}

Consider the algebra of classical brackets on the coordinate algebra
$\complex[\vec\xi \ ]$ which is isomorphic to the commutator algebra of imaginary octonions (\ref{oct2}),
\begin{equation}\label{oct4}
\{\xi_A,\xi_B\}_\eta=2\, \eta_{ABC}\, \xi_C \ ,
\end{equation}
where $\vec\xi=(\xi_A)$ with $\xi_A\in\real$, $A=1,\dots,7$. This bracket is bilinear,
antisymmetric and satisfies the Leibniz rule by definition. Note that although the algebra (\ref{oct4}) is isomorphic to the Malcev algebra of octonions (\ref{oct2}) it does not satisfy Malcev identity, \cite{KS17}. Moreover, in \cite{Vassilevich:2018gkl} it was shown more stronger statement: the quasi-Poisson structure satisfies the Malcev identity only if it is Poisson, i.e., satisfies the Jacobi identity. 

Introducing $\sigma^i:=\xi_{i+3}$ for
$i=1,2,3$ and $\sigma^4:= \xi_7$, one may rewrite (\ref{oct4}) in components as
\begin{eqnarray}\label{mp1}
\{\xi_i,\xi_j\}_\eta&=&2\, \varepsilon_{ijk}\, \xi_k \qquad \mbox{and} \qquad
                   \{\sigma^4,\xi_i\}_\eta \ = \ 2\, \sigma^i \ , \\[4pt]
\{\sigma^i,\sigma^j\}_\eta&=&-2\, \varepsilon^{ijk}\, \xi_k \qquad \mbox{and} \qquad
                     \{\sigma^4,\sigma^i\}_\eta \ = \ -2\, \xi_i \ , \nonumber\\[4pt]
\{\sigma^i,\xi_j\}_\eta&=&-2\,( \delta_j^i\, \sigma^4- \varepsilon^{i}{}_{jk}\,
                         \sigma^k) \ . \nonumber
\end{eqnarray}
Using \eqref{oct3} the non-vanishing Jacobiators can be written as
\begin{eqnarray}\label{mp2}
\{\xi_i,\xi_j,\sigma^k\}_\eta&=&-4\, (\varepsilon_{ij}{}^{k}\,
                                \sigma^4+ \delta^{k}_{j}\,
                          \sigma_i- \delta^{k}_{i}\, \sigma_j ) \ ,\\[4pt]
\{\xi_i,\sigma^j,\sigma^k\}_\eta&=&4\, ( \delta_{i}^{j}\, \xi_k-\delta_{i}^{k}\, \xi_j)
                           \ ,\nonumber\\[4pt]
\{\sigma^i,\sigma^j,\sigma^k\}_\eta&=&4\, \varepsilon^{ijk}\, \sigma^4 \ ,\nonumber\\[4pt]
\{\xi_i,\xi_j,\sigma^4\}_\eta&=&4\, \varepsilon_{ijk}\, \sigma^k \ ,\nonumber\\[4pt]
\{\xi_i,\sigma^j,\sigma^4\}_\eta&=&4\, \varepsilon_{i}{}^{jk}\, \xi_k \ ,\nonumber\\[4pt]
\{\sigma^i,\sigma^j,\sigma^4\}_\eta&=&-4\, \varepsilon^{ijk}\, \sigma^k \ . \nonumber
\end{eqnarray}

The calculation of the first orders of the generalized Bopp shift indicates the ansatz:
\begin{equation}
\xi_A=y_A-\eta_{ABC}\,\pi_B\,y_C-\left(y_A\,\pi^2-\pi_A\,y_B\pi_B\right)\,\chi(\pi^2)\,,\label{Bopp-eta}
\end{equation}
where $\chi(\pi^2)$ is a function to be determined. One calculates
\begin{equation}
\{\xi_A,\xi_B\}-2\,\eta_{ABC}\,\xi_C=(y_A\,\pi_B-y_B\,\pi_A)\,\left[2\pi^2\chi^\prime+3\chi-1-\pi^2\chi^2\right]+4\,\eta_{ABCD}\,\pi_C\,y_D\,.
\end{equation}
From where we get an equation on $\chi(t)$:
\begin{equation}
2t\chi^\prime+3\chi-1-t\chi^2=0\,,\qquad \chi(0)=\frac13\,,
\end{equation}
with the solution
\begin{equation}
\chi(t)=-\frac1t\,\left(\sqrt{t}\cot\sqrt{t}-1\right)\,.
\end{equation}

The perturbative calculation for $\omega_{AB}\,(\xi,\tilde\xi)$ suggests the ansatz:
\begin{equation}
\omega_{AB}\,(\xi,\tilde\xi)=2\,\eta_{ABC}\,\xi_C+\eta_{ABCD}\,\tilde\xi_C\,\xi_D\ \phi\left(\tilde\xi^2\right)+\eta_{ABCD}\,\eta_{DEF}\,\tilde\xi_C\,\tilde\xi_E\,\xi_F\ \psi\left(\tilde\xi^2\right)\,,
\end{equation}
where the functions $\phi\left(\tilde\xi^2\right)$ and $\psi\left(\tilde\xi^2\right)$ can be found from the relation
\begin{equation}
4\,\eta_{ABCD}\,\pi_C\,y_D=\eta_{ABCD}\,\tilde\xi_C\,\xi_D\ \phi\left(\tilde\xi^2\right)+\eta_{ABCD}\,\eta_{DEF}\,\tilde\xi_C\,\tilde\xi_E\,\xi_F\ \psi\left(\tilde\xi^2\right)\,.
\end{equation}
Using the expression for the generalized Bopp shift (\ref{Bopp-eta}), as well as the choice $\tilde\xi=\pi$, contraction identity (\ref{epsilon7}) and the antisymmetry of $\eta_{ABCD}$ we find the relations
\begin{eqnarray*}
-\phi+ \psi\,\left(1-\pi^2\,\chi(\pi^2)\right)&=&0\,,\\
\phi\,\left(1-\pi^2\,\chi(\pi^2)\right)+\pi^2\,\psi&=&4\,,
\end{eqnarray*}
implying that
\begin{equation}
\phi=2\,\frac{\sin 2\sqrt{\pi^2}}{\sqrt{\pi^2}}\,,\qquad \mbox{and}\qquad  \psi=4\,\frac{\sin^2\sqrt{\pi^2}}{\pi^2}\,.
\end{equation}

Finally we conclude that the symplectic realization of the quasi-Poisson structure (\ref{oct4}) is given by
\begin{eqnarray}
\{\xi_A,\xi_B\}_p&=&2\,\eta_{ABC}\,\xi_C+2\,\frac{\sin 2\sqrt{\tilde\xi^2}}{\sqrt{\tilde\xi^2}}\,\eta_{ABCD}\,\tilde\xi_C\,\xi_D+4\,\frac{\sin^2\sqrt{\tilde\xi^2}}{\tilde\xi^2}\,\eta_{ABCD}\,\eta_{DEF}\,\tilde\xi_C\,\tilde\xi_E\,\xi_F\,,\nonumber\\
\{\xi_A,\tilde\xi_B\}_p&=&\delta_{AB}+\eta_{ABC}\,\tilde\xi_C+\left(\delta_{AB}\,\tilde\xi^2-\tilde\xi_A\,\tilde\xi_B\right)\,\frac{\sqrt{\tilde\xi^2}\,\cot\sqrt{\tilde\xi^2}-1}{\tilde\xi^2}\,,\nonumber\\
\{\tilde\xi_A,\tilde\xi_B\}_p&=&0\,.\label{sr-eta}
\end{eqnarray}
while the expression for the generalized Bopp shift is
\begin{eqnarray}
\xi_A&=&y_A-\eta_{ABC}\,\pi_B\,y_C+\left(y_A\,\pi^2-\pi_A\,y_B\pi_B\right)\,\frac{\sqrt{\pi^2}\,\cot\sqrt{\pi^2}-1}{\pi^2}\,,\label{bs-oct}\\
\tilde\xi_A&=&\pi_A\,.\nonumber
\end{eqnarray}

Note that the restriction of the restriction of the quasi-Poisson structure (\ref{oct4}) to the three-dimensional space with coordinates $\xi_i$, $i=1,2,3$, results according to (\ref{mp1}), in the Poisson structure $\{\xi_i,\xi_j\}_\varepsilon=2\, \varepsilon_{ijk}\,\xi_k$, isomorphic to the $su(2)$ Lie algebra. Since in three dimensions  the totally antisymmetric tensor $\eta_{ABCD}$ of the rank four automatically vanishes, from (\ref{sr-eta}) one obtains immediately the symplectic realization of the $su(2)$-like Poisson structure,
\begin{eqnarray}
\{\xi_i,\xi_j\}_p&=&2\, \varepsilon_{ijk}\,\xi_k\,,\nonumber\\
\{\xi_i,\tilde\xi_j\}_p&=&\delta_{ij}+\varepsilon_{ijk}\,\tilde\xi_k+\left(\delta_{ij}\,\tilde\xi^2-\tilde\xi_i\,\tilde\xi_j\right)\,\frac{\sqrt{\tilde\xi^2}\,\cot\sqrt{\tilde\xi^2}-1}{\tilde\xi^2}\,,\nonumber\\
\{\tilde\xi_i,\tilde\xi_j\}_p&=&0\,.\label{sr-epsilon}
\end{eqnarray}
The corresponding Bopp shift is given by,
\begin{equation}
\xi_i=y_i-\varepsilon_{ijk}\,\pi_j\,y_k+\left(y_i\,\pi^2-\pi_i\,y_j\pi_j\right)\,\frac{\sqrt{\pi^2}\,\cot\sqrt{\pi^2}-1}{\pi^2}\,,\qquad \tilde\xi_i=\pi_i\,.\label{bs-quat}
\end{equation}
Substituting formally in the above expression the canonical phase space momenta $\pi_i$ by the derivative operators, $-\ii\partial_i$, the canonical phase space coordinates $y_i$ by the multiplication operators $x_i$, and supposing the normal ordering, one recovers the expression for the operators $\hat\xi_i$, giving the polydifferential representation of the algebra, $[\hat\xi_i,\hat\xi_j]=\ii\, \varepsilon_{ijk}\,\hat\xi_k$, which was obtained in \cite{KupVit} and used for the derivation of the $su(2)$-like star product.

\subsection{R-flux in M-theory}

To relate the quasi-Poisson structure (\ref{oct4}) to the constant $R$-flux algebra (\ref{rb}) let us following \cite{GLM} introduce the $7\times7$ matrix 
\begin{equation}\label{oct6}
{\mit\Lambda} = \big({\mit\Lambda}^{AB}\big)= \frac1{2\hbar} \,
\begin{pmatrix}
0 & {\sqrt{\lambda\, \ell_s^3\, R}}\ \unit_3 & 0 \\
0 & 0 & \sqrt{\lambda^3\, \ell_s^3\, R} \\
 -\lambda\, \hbar\ \unit_3 & 0 & 0
\end{pmatrix}
\end{equation}
with $\unit_3$ the $3\times3$ identity matrix. The matrix $\mit\Lambda$
is non-degenerate as long as all parameters are non-zero, but it is
not orthogonal. Using it we define new coordinates
\begin{equation}\label{oct8}
\vec x = \big(x^A\big) = \big(\mbf x,x^4,\mbf p\big) := {\mit\Lambda}\,
\vec\xi= \mbox{$\frac1{2\hbar}$}\, \big(\sqrt{\lambda\,
      \ell_s^{3}\, R}\ \mbf\sigma\,,\,
\sqrt{\lambda^3\, \ell_s^{3}\, R}\ \sigma^4\,,\, -\lambda\,\hbar\ \mbf\xi \big) \ .
\end{equation}
From the classical brackets (\ref{oct4}) one obtains the quasi-Poisson algebra
\begin{equation}\label{oct4a}
\{x^A,x^B\}_\lambda=2\,\lambda^{ABC} \, x^C \qquad \mbox{with} \quad \lambda^{ABC}:= {\mit\Lambda}^{AA^\prime}\,{\mit\Lambda}^{BB^\prime}\,\eta_{A^\prime B^\prime C^\prime}\, {\mit\Lambda}^{-1}_{C^\prime C} \ ,
\end{equation}
which can be written in components as
\begin{eqnarray}\label{oct9}
\{x^i,x^j\}_\lambda &=&\mbox{$\frac{\ell_s^3}{\hbar^2}$}\,
                        R^{4,ijk4}\, p_k \qquad \mbox{and} \qquad
                        \{x^4,x^i\}_\lambda \ = \ \mbox{$\frac{\lambda\, \ell_s^3}{\hbar^2}$}\, R^{4,1234}\, p^i \ , \\[4pt]
\{x^i,p_j\}_\lambda &=&\delta^i_j\,x^4+\lambda\,
                        \varepsilon^i{}_{jk}\, x^k \qquad \mbox{and}
                        \qquad \{x^4,p_i\}_\lambda \ = \ \lambda^2\,x_i \ , \nonumber\\[4pt]
\{p_i,p_j\}_\lambda &=&-\lambda\, \varepsilon_{ijk}\, p^k\ . \nonumber
\end{eqnarray}
Now taking the contraction limit $\lambda\to0$ we observe that the element $x^4$ becomes a central element and can be taken to be identity, while the phase space coordinates $x^i$ and $p_i$ form the algebra of the constant $R$-flux algebra (\ref{rb}). The main conjecture of \cite{GLM} is that the quasi-Poisson brackets (\ref{oct9}) provide the uplift of the string $R$-flux algebra to M-theory. In this sense $\lambda$ plays the role of the M-theory radius.

The corresponding Jacobiators are
\bea\nonumber
\{x^A,x^B,x^C\}_\lambda = -4\, \lambda^{ABCD} \, x^D \qquad \mbox{with} \quad \lambda^{ABCD}:= {\mit\Lambda}^{AA^\prime}\,{\mit\Lambda}^{BB^\prime}\,{\mit\Lambda}^{CC^\prime}\,\eta_{A^\prime B^\prime C^\prime D^\prime}\, {\mit\Lambda}^{-1}_{D^\prime D} \ ,
\eea
with the components
\bea\label{oct9a}
\{x^i,x^j,x^k\}_\lambda &=& \mbox{$\frac{\ell_s^3}{\hbar^2}$}\, R^{4,ijk4} \, x^4 \ , \\[4pt]
\{x^i,x^j,x^4\}_\lambda &=& -\mbox{$\frac{\lambda^2\, \ell_s^3}{\hbar^2}$}\, R^{4,ijk4} \, x_k \ , \nonumber \\[4pt]
\{p_i,x^j,x^k\}_\lambda &=& \mbox{$\frac{\lambda\,\ell_s^3}{\hbar^2}$} \, R^{4,1234}\, \big(\delta^j_i\, p^k-\delta^k_i\, p^j \big) \ , \nonumber \\[4pt]
\{p_i,x^j,x^4\}_\lambda &=& \mbox{$\frac{\lambda^2\, \ell_s^3}{\hbar^2}$}\, R^{4,ijk4}\, p_k \ , \nonumber \\[4pt]
\{p_i,p_j,x^k\}_\lambda &=& - \lambda^2\, \varepsilon_{ij}{}^{k}\, x^4-\lambda\, \big(\delta_j^k\, x_i-\delta_i^k\, x_j \big) \ , \nonumber \\[4pt]
\{p_i,p_j,x^4\}_\lambda &=& \lambda^3\, \varepsilon_{ijk}\, x^k \ , \nonumber \\[4pt]
\{p_i,p_j,p_k\}_\lambda &=& 0 \ . \nonumber
\eea

Making transformation (\ref{oct8}) in (\ref{sr-eta}) and (\ref{bs-oct}) we obtain:
\begin{eqnarray}
\{x_A,x_B\}_p&=&2\,\lambda_{ABC}\,x_C+2\,\frac{\sin 2\sqrt{(\Lambda\tilde x)^2}}{\sqrt{(\Lambda\tilde x)^2}}\,\lambda_{ABCD}\,\tilde x_C\,x_D+4\,\frac{\sin^2\sqrt{(\Lambda\tilde x)^2}}{(\Lambda\tilde x)^2}\,\lambda_{ABCD}\,\lambda_{DEF}\,\tilde x_C\,\tilde x_E\,x_F\,,\nonumber\\
\{x_A,\tilde x_B\}_p&=&\delta_{AB}+\lambda_{ABC}\,\tilde x_C+\left(\delta_{AB}\,(\Lambda\tilde x)^2-\Lambda_{AA^\prime}\,\tilde x_{A^\prime}\,\Lambda_{BB^\prime}\,\tilde x_{B^\prime}\right)\,\frac{\sqrt{(\Lambda\tilde x)^2}\,\cot\sqrt{(\Lambda\tilde x)^2}-1}{(\Lambda\tilde x)^2}\,,\nonumber\\
\{\tilde x_A,\tilde x_B\}_p&=&0\,.\label{sr-RM}
\end{eqnarray}
While the expression for the generalized Bopp shift reads
\begin{eqnarray}
x_A&=&y_A-\lambda_{ABC}\,\pi_B\,y_C+\left(y_A\,(\Lambda\,\pi)^2-\Lambda_{AA^\prime}\,\pi_{A^\prime}\,\Lambda_{BB^\prime}\,\pi_{B^\prime}\right)\,\frac{\sqrt{(\Lambda\,\pi)^2}\,\cot\sqrt{(\Lambda\,\pi)^2}-1}{(\Lambda\,\pi)^2}\,,\nonumber\\
\tilde x_A&=&\pi_A\,.\label{bs-RM}
\end{eqnarray}

Following the logic of \cite{BaLu} that the magnetic monopole algebra can be obtained from the non-geometric $R$-flux quasi-Poisson structure by swapping the phase space coordinates and momenta, in \cite{LMS} it was introduced a magnetic analogue of the M-theory $R$ flux algebra (\ref{oct9}), the smeared Kaluza-Klein monopole. In \cite{KS18} we used the symplectic realization of the monopole algebra to study the classical dynamics and quantization of the electric charge in a field of the magnetic pole distributions. So, it is reasonable to continue this logic and use the symplectic realization (\ref{sr-RM}) for studying the smeared Kaluza-Klein monopole.

\subsection*{Acknowledgments}

I am grateful to Dima Vassilevich, Richard Szabo and Patrizia Vitale for helpful discussions.
This work was supported by the Grant 305372/2016-5 from the Conselho Nacional de Pesquisa (CNPq, Brazil), and the Capes-Humboldt Fellowship 0079/16-2. 

\appendix

\newsection{Octonions}

The algebra $\mathbb{O}$ of octonions is the best known example of a nonassociative but alternative algebra. Every octonion $X\in \mathbb{O}$ can be written in the form
\begin{equation}\label{oct}
X=k^0\, \unit+ k^A\,e_A
\end{equation}
where $k^0,k^A\in\real$, $A=1,\dots,7$, while $\unit$ is the identity element and the imaginary unit octonions $e_A$ satisfy the multiplication law
\begin{equation}\label{oct1}
e_A\, e_B=-\delta_{AB}\, \unit +\eta_{ABC}\, e_C \ .
\end{equation}
Here $\eta_{ABC}$ is a completely antisymmetric tensor of rank three
with nonvanishing values
\bea
\eta_{ABC}=+1 \qquad \mbox{for} \quad ABC = 123 , \ 435, \ 471, \ 516,
\ 572, \ 624, \ 673 \ .
\label{eq:etadef}\eea
Introducing $f_i:=e_{i+3}$ for $i=1,2,3$, the algebra (\ref{oct1}) can be rewritten as
\begin{eqnarray}\label{oct1a}
e_i\, e_j&=&-\delta_{ij}\, \unit +\varepsilon_{ijk}\, e_k\ ,\\[4pt]
e_i\, f_j&=&\delta_{ij}\, e_7-\varepsilon_{ijk}\, f_k\ ,\nonumber\\[4pt]
f_i\, f_j&=&\delta_{ij}\, \unit -\varepsilon_{ijk}\, e_k\ ,\nonumber\\[4pt]
e_7\, e_i&=&f_i \qquad \mbox{and} \qquad f_i\, e_7 \ = \ e_i\ , \nonumber
\end{eqnarray}
which emphasises a subalgebra $\quat$ of quaternions generated by
$e_i$; we will use this component form of the algebra $\mathbb{O}$ frequently in what follows.

The algebra $\mathbb{O}$ is neither commutative nor associative. The commutator algebra of the octonions is given by
\begin{equation}\label{oct2}
[e_A,e_B]:=e_A\, e_B-e_B\, e_A=2\, \eta_{ABC}\, e_C\ ,
\end{equation}
which can be written in components as
\begin{eqnarray}\label{oct2a}
[e_i,e_j]&=&2\, \varepsilon_{ijk}\, e_k \qquad \mbox{and} \qquad [e_7,e_i] \ = \ 2\,
             f_i\ ,\\[4pt]
[f_i,f_j]&=&-2\, \varepsilon_{ijk}\, e_k \qquad \mbox{and} \qquad [e_7,f_i] \ = \
             -2\, e_i\ ,\nonumber\\[4pt]
[e_i,f_j]&=&2\, (\delta_{ij}\, e_7- \varepsilon_{ijk}\, f_k) \ .\nonumber
\end{eqnarray}
The structure constants $\eta_{ABC}$ satisfy the contraction identity
\begin{equation}\label{epsilon7}
\eta_{ABC}\, \eta_{DEC}=\delta_{AD}\, \delta_{BE}-\delta_{AE}\,
\delta_{BD}+\eta_{ABDE} \ ,
\end{equation}
where $\eta_{ABDE}$ is a completely antisymmetric tensor of rank four
with nonvanishing values
$$
\eta_{ABDE}= +1 \qquad \mbox{for} \quad ABDE = 1267, \ 1346, \ 1425, \
1537, \ 3247, \ 3256, \ 4567 \ .
$$
One may also represent the rank four tensor
$\eta_{ABDE}$ as the
dual of the rank three tensor $\eta_{FGH}$ through 
\begin{equation}\label{epsilon8}
\eta_{ABDE}=\mbox{$\frac{1}{6}$}\, \varepsilon_{ABDEFGH}\, \eta_{FGH}
\ ,
\end{equation}
where $\varepsilon_{ABDEFGH}$ is the alternating symbol in seven
dimensions normalized as $\varepsilon_{1234567}=+1$. Together they satisfy the contraction identity
\bea
\eta_{AEF}\, \eta_{ABCD} &=& \delta_{EB}\, \eta_{FCD}-\delta_{FB}\, \eta_{ECD} +\delta_{EC}\, \eta_{BFD}-\delta_{FC}\, \eta_{BED} \nonumber \\ && +\, \delta_{ED}\, \eta_{BCF}-\delta_{FD}\, \eta_{BCE} \ .
\label{eq:eta34}\eea
Taking into account (\ref{epsilon7}), for the Jacobiator we get
\begin{equation}\label{oct3}
[e_A,e_B,e_C]:=\mbox{$\frac{1}3$} \,\left([e_A,[e_B,e_C]]+[e_C,[e_A,e_B]]+[e_B,[e_C,e_A]]\right)=-4\,
\eta_{ABCD}\, e_D \ ,
\end{equation}
and the alternative property of the algebra $\mathbb{O}$ implies that the
Jacobiator is proportional to the associator, i.e.,
$[X,Y,Z]=6\,\big((X\,Y)\, Z-X\,(Y\, Z) \big)$ for any three octonions
$X,Y,Z\in\mathbb{O}$.

\end{document}